\documentclass[11pt,letterpaper]{article}
\usepackage[margin=1in,letterpaper]{geometry}
\usepackage[utf8]{inputenc}
\usepackage{pslatex}  % use times roman
\usepackage{amsmath}
\usepackage{amssymb}
\usepackage{amsthm}
\usepackage{hyperref}
\usepackage{url}
\usepackage{cite}

\newcommand{\defn}[1]{\emph{\textbf{{#1}}}}

\newcommand{\NN}{\mathbb{N}}
\newcommand{\poly}{\operatorname{poly}}
\newcommand{\polylog}{\operatorname{polylog}}
\newcommand{\E}{\mathbb{E}}

\newcommand{\tow}{\operatorname{tow}}
\newcommand{\si}{\mathbf{S}}
\newcommand{\ham}{\operatorname{Ham}}

\renewcommand{\paragraph}[1]{\vspace{.5 cm} \noindent \textbf{#1.} }

\usepackage{thmtools}
\usepackage{thm-restate}

\newtheoremstyle{slanted}% <name>
{3pt}% <Space above>
{3pt}% <Space below>
{\slshape}% <Body font>
{}% <Indent amount>
{\bfseries}% <Theorem head font>
{.}% <Punctuation after theorem head>
{.5em}% <Space after theorem heading>
{}% <Theorem head spec (can be left empty, meaning `normal')>

\theoremstyle{slanted}
\newtheorem{theorem}{Theorem}%[section]
\newtheorem{lemma}[theorem]{Lemma}

\newtheorem{proposition}[theorem]{Proposition}
\newtheorem{definition}[theorem]{Definition}

\begin{document}

\title{Memoryless Worker-Task Assignment with Polylogarithmic Switching Cost}

\author{Aaron Berger\footnote{Supported by NSF Graduate Research Fellowship Program DGE-1745302.} \\MIT\and William Kuszmaul\footnote{Supported by an NSF GRFP fellowship and a Fannie and John Hertz Fellowship. Research was partially sponsored by the United States Air Force Research Laboratory and was accomplished under Cooperative Agreement Number FA8750-19-2-1000. The views and conclusions contained in this document are those of the authors and should not be interpreted as representing the official policies, either expressed or implied, of the United States Air Force or the U.S. Government. The U.S. Government is authorized to reproduce and distribute reprints for Government purposes notwithstanding any copyright notation herein.}\\ MIT\and Adam Polak\footnote{supported by the Swiss National Science Foundation within the project \emph{Lattice Algorithms and Integer Programming} (185030). Part of this work was done at Jagiellonian University, supported by the National Science Center of Poland grant 2017/27/N/ST6/01334.}\\ EPFL\and Jonathan Tidor\footnote{Supported by NSF Graduate Research Fellowship Program DGE-1745302.}\\ MIT\and Nicole Wein\footnote{This work was done at MIT, supported by NSF Grant CCF-1514339.}\\ DIMACS}

\date{}

\maketitle

\begin{abstract}
  We study the basic problem of assigning memoryless workers to tasks with dynamically changing demands. Given a set of $w$ workers and a multiset $T \subseteq[t]$ of $|T|=w$ tasks, a memoryless worker-task assignment function is any function $\phi$ that assigns the workers $[w]$ to the tasks $T$ based only on the current value of $T$. The assignment function $\phi$ is said to have switching cost at most $k$ if, for
  every task multiset $T$, changing the contents of $T$ by one task
  changes $\phi(T)$ by at most $k$ worker assignments. The goal of memoryless worker task assignment is to construct an assignment function with the smallest possible switching cost.
  
  In past work, the problem of determining the optimal switching cost has been posed as an open question. There are no known sub-linear upper bounds, and after considerable effort, the best known lower bound remains 4 (ICALP 2020). 
  
  We show that it is possible to achieve polylogarithmic switching cost. We give a construction via the probabilistic method that achieves
  switching cost $O(\log w \log (wt))$ and an explicit construction that
  achieves switching cost $\polylog (wt)$. We also prove a super-constant lower bound on switching cost: we show that for any value of $w$, there exists
  a value of $t$ for which the optimal switching cost is $w$. Thus it is not possible to achieve a switching cost that is sublinear \emph{strictly} as a function of $w$. 

  Finally, we present an application of the worker-task assignment problem to a metric embeddings problem. In particular, we use our results to give the first low-distortion embedding from sparse binary vectors into low-dimensional Hamming space.
\end{abstract}

\vfill
\pagebreak

\section{Introduction}

The general problem of \defn{distributed task allocation}, where a group of agents must collectively allocate themselves to tasks, has been studied in a wide variety of settings where the agents have varying degrees of communication, memory, knowledge of the system, and faultiness (see Georgiou and Shvartsman's book~\cite{georgiou2011cooperative} for a survey). In this paper we are interested in the dynamic version of a distributed task allocation problem, that is, where the demands for each task are changing over time. Dynamic task allocation has been the focus of a great deal of both empirical and theoretical work in areas such as swarm robotics \cite{krieger2000ant, csahin2004swarm, mclurkin2005dynamic, mclurkin2004stupid, lerman2006analysis, macarthur2011distributed, su2020lower} and collective insect behavior \cite{beshers2001models, robinson1992regulation, radeva2017costs, su2017ant, su2020lower}.

Although there are many possible variations of the dynamic task allocation problem (in particular, in terms of what the capabilities of the workers are), most share a basic common structure. For positive integers $w$ and $t$, there are $w$ workers $1, 2, \ldots, w$, and there is a multiset $T \subseteq [t]$ of $|T|=w$ tasks.\footnote{If a task $j$ appears $m_T(j)$ times in the multiset $T$, then one should think of the task as having a current demand of $m_T(j)$ workers. By writing $T \subseteq [t]$ we mean that the elements of multiset $T$ are from $\{1, 2, \ldots, t\}$, but the multiplicity of each element can be arbitrarily large.} The multiset $T$ changes gradually over time: in each time step, one new task is added and one old task is removed. The goal is to maintain a dynamic assignment of workers $\{1, 2, \ldots, w\}$ to tasks $T$ such that the \defn{switching cost} is as small as possible, i.e., the number of worker-task assignments that change each time that $T$ changes is bounded by some small quantity.

The goal of studying dynamic distributed task allocation is to answer the following question: to what degrees are various capabilities (i.e., memory, communication, knowledge of the system, computational power, etc.) needed for the workers to guarantee a small switching cost? 

Motivated in part by applications to swarm robotics and collective ant behavior, recent work \cite{su2017ant, su2020lower} has focused on the question of what happens when workers are \emph{completely memoryless}. At any given moment each worker $i \in \{1, 2, \ldots, w\}$ must determine which task $\tau \in T$ they are assigned to solely as a function of the current task multiset $ T$; the workers do not remember anything about the past system states or worker assignments.  Note that this memoryless-ness requirement is sufficiently strong that it also implies communicationless-ness---indeed, if workers cannot remember where they were in previous steps, and all that each worker knows is the current multiset of tasks, then there is no worker-specific information to be shared. Being memoryless and communicationless is especially important to settings where a worker might suffer a fault and thus ``reboot''; in this case, the worker can determine which task they are assigned to without relying on what was stored in its memory and without knowledge of which tasks other workers are assigned to at the moment \cite{su2020lower}.

It remains an open question \cite{su2020lower, su2017ant} whether memoryless workers can achieve even a \emph{sub-linear} switching cost. In this paper we show that, not only is sub-linear switching cost possible, but it is even possible to achieve a polylogarithmic switching cost of $\polylog (wt)$. We also prove a lower bound for any algorithm that wishes to parameterize by only $w$ and not $t$: for any $w$, if $t \gg w$ is sufficiently large, then the trivial switching cost of $w$ becomes optimal.

\paragraph{Formal problem statement}
Throughout this paper, we shall refer to the memoryless dynamic distributed task allocation problem simply as the \defn{worker-task assignment problem}. Formally, the worker-task assignment problem is defined as follows. There are $w$ workers $1, 2, \ldots, w$ and $t$ tasks $1, 2, \ldots, t$. A \defn{worker-task assignment function} $\phi$ is a function that takes as input a multiset $T$ of $w$ tasks, and produces an assignment of workers to tasks such that
the number of workers assigned to a given task $\tau \in T$ is equal to the multiplicity of $\tau$ in $T$.

Two task multisets $T_1, T_2$ of size $w$ are said to be \emph{adjacent} if they agree on exactly $w-1$ elements; that is, $|T_1\setminus T_2|=|T_2\setminus T_1|=1$.\footnote{Let $m_A(i)$ denote the number of times element $i$ appears in multiset $A$. Then, for any two multisets $A$~and $B$, we define multisets $A \setminus B$, $A \cup B$, and $A \cap B$ to be such that $m_{A\setminus B}(i) = \max(0, m_A(i)-m_B(i))$, $m_{A \cup B}(i) = \max(m_A(i), m_B(i))$, and $m_{A \cap B}(i) = \min(m_A(i), m_B(i))$, for every element $i$.} The switching cost between two adjacent task multisets $T_1, T_2$ of size $w$ is defined as the number of workers whose assignment changes between $\phi(T_1)$ and $\phi(T_2)$. The \defn{switching cost of $\phi$} is defined to be the maximum switching cost over \emph{all} pairs of adjacent task multisets. The goal of the worker-task assignment problem is to design a worker-task assignment function with the minimum possible switching cost.

\paragraph{Prior work in the memoryless setting}
The optimal switching cost is trivially between $1$ and $w$. Improving either of these bounds substantially has proven difficult, however. 
 
Su, Su, Dornhaus, and Lynch~\cite{su2017ant} initiated the study of the worker-task assignment problem and observed that assigning the workers to tasks in numerical order achieves a switching cost of $\min(t-1, w)$. They also proved a lower bound of $2$ on switching cost, and showed a matching upper bound in the case where $w \le 6$ and $t \le 4$. 

Subsequent work by Su and Wein~\cite{su2020lower}, in ICALP 2020, pushed further on the lower-bound side of the problem. They proved that a switching cost of 2 is not always possible in general. They show that, if $t\geq 5$ and $w\geq 3$, then any worker-task assignment function must have switching cost at least $3$; and if $t$ is sufficiently large in terms of $w$ (i.e., it is a tower of height $w-1$), then the switching cost must be at least $4$. 

The bounds by \cite{su2017ant} and \cite{su2020lower} have until now remained state-of-the-art. It remains unknown whether the optimal switching cost is small (it could be as small as $4$) or large (it could be as large as $\min(t-1, w)$). And even achieving a lower bound of $4$ on switching cost \cite{su2020lower} has required a quite involved argument.

\subsection{This paper}
The contributions of this paper are twofold. First, we present significant progress on the worker-task assignment problem, resulting in both a polylogarithmic upper bound and a super-constant lower bound for the optimal switching cost---our results are interesting in part because of their use of randomized techniques to construct a deterministic assignment function. Second, we explore a natural connection between the worker-task assignment problem and the metric embedding problem of densification into Hamming space, and we transform our progress on the former into new results for the latter.

\paragraph{Results on worker-task assignment}
Our first result establishes that it is possible to construct a worker-task assignment function with $O(\log w \log (wt))$ switching cost. This resolves the open question as to whether memoryless worker can allocate themselves to tasks with strong worst-case guarantees.
\begin{restatable}{theorem}{thmUB}
  There exists a worker-task assignment function that achieves switching cost $O(\log w \log (wt))$.
  \label{thm:UB}
\end{restatable}

Theorem~\ref{thm:UB} is proven via the probabilistic method and is thus non-constructive. By replacing random hash functions with strong dispersers, however, we show that one can construct an explicit worker-task assignment function with polylogarithmic switching cost.

\begin{restatable}{theorem}{thmUBcon}
  There is an explicit worker-task assignment function that achieves switching cost $O(\polylog (wt))$.
  \label{thm:UBcon}
  \end{restatable}

Both Theorems \ref{thm:UB} and \ref{thm:UBcon} continue to hold in the more general setting where the size of $T$ changes over time. That is, $T$ is permitted to be any multiset of $[t]$ of size $w$ or smaller. Two task multisets $T_1, T_2$ of different sizes are considered adjacent if they satisfy $\big| |T_1| - |T_2| \big| = 1$ and $|(T_1 \cup T_2) \setminus (T_1 \cap T_2)| = 1$. If $|T| < w$, then our worker-task assignment function assign workers $1, \ldots, |T|$ to tasks, and leaves workers $|T| + 1, \ldots, w$ unassigned. 

Finally, from the lower bounds side, we prove that no algorithm can achieve sub-linear switching cost as a function of only $w$. Theorem~\ref{thm:LB} says that, if $t$ is sufficiently larger than $w$, then for any worker-task assignment function, there must exist a pair of adjacent task multi-sets that forces \emph{all} of the $w$ workers to be reassigned. In the statement of the theorem, and throughout the paper, $\tow(n)$ is defined to be a tower of twos of height $n$ (i.e., the inverse of the $\log^*$ function).

\begin{restatable}{theorem}{thmLB}
\label{thm:LB}
  For every $w$ and $t\geq\tow(\Omega(w))$, every worker-task assignment function has switching cost $w$.
\end{restatable}

This represents the first super-constant lower bound for the switching cost of a worker-task assignment function. Another way to think about the theorem is that for every $t$, there is some $w$ for which any worker-task assignment function has switching cost at least $\Omega(\log^*(t))$. Therefore our bounds leave a gap between $\log^*$ and $\polylog$ in terms of dependence on $t$.

\paragraph{An application to metric embeddings: Densification into Hamming space}
The problem of embedding one metric space $\mathcal{M}_1$ into another
metric space $\mathcal{M}_2$ with small distortion has been widely
studied in many contexts and has found many algorithmic applications
\cite{ostrovsky2007low, chakraborty2016streaming,
  charikar2006embedding,charikar2018estimating,fakcharoenphol2004tight,
  bansal2011polylogarithmic, brinkman2005impossibility,
  charikar2002dimension,
  andoni2011near,bourgain1985lipschitz,johnson1984extensions,matouvsek2008variants,
  berinde2008combining, linial1995geometry}.

Bourgain \cite{bourgain1985lipschitz} initiated the study of metric
embeddings (into normed spaces) by showing that
$O(\log |M|)$-distortion embeddings into $\ell_2$ are possible for any
space $M$. Much of the subsequent work has focused either on
embeddings between exponentially large metric spaces \cite{ostrovsky2007low,
  chakraborty2016streaming,
  charikar2006embedding,charikar2018estimating,fakcharoenphol2004tight,
  bansal2011polylogarithmic, charikar2002dimension,
  berinde2008combining}, or on embeddings with sub-logarithmic
distortion \cite{johnson1984extensions,matouvsek2008variants,
  berinde2008combining, charikar2002dimension}.

One natural question is that of \defn{densification}: can one embed
sparse vectors from a high-dimensional $\ell_1$-space into
a low-dimensional $\ell_1$-space? That is, if $V^k_n$
is the set of $n$-dimensional vectors with $k$ non-zero entries, what
is the smallest $m$ for which $V^k_n$ can be embedded into $m$-dimensional $\ell_1$-space with low distortion?  Charikar and Sahai \cite{charikar2002dimension}
were the first to consider this problem, and showed how to achieve an
output dimension of $m = O((k / \epsilon)^2 \log n)$ with distortion
$1 + \epsilon$. They also showed how to apply densification to the
related problem of embedding arbitrary tree metrics into
low-dimensional $\ell_1$-space
\cite{charikar2002dimension}. Subsequently, Berinde et
al. \cite{berinde2008combining} used expander graphs in order to
achieve $m = O(k \log (n / k) / \epsilon^2)$ with distortion
$1 + \epsilon$. They then used their densification embedding as a tool
to perform sparse signal recovery \cite{berinde2008combining,
  gilbert2010sparse, gilbert2006algorithmic,indyk2008explicit}. Both
of the known densification algorithms \cite{charikar2002dimension,
  berinde2008combining} rely on linear sketches, in which each vector
$\vec{x} \in V^k_n$ is mapped to a vector of the form
$\sum_{i} x_i \vec{b_i}$ for some set of vectors
$\vec{b_1}, \ldots, \vec{b_n}$.

The prior work on densification \cite{charikar2002dimension,
  berinde2008combining} has focused on embedding into $\ell_1$-space.
  In
Section \ref{sec:metric}, we consider the same problem over Hamming space, where the
distance between two vectors $\vec{x}, \vec{y}$ is given by
$\ham(\vec{x}, \vec{y}) = |\{i \mid x_i \neq y_i\}|$. Densification over Hamming space requires
new techniques due to the fact that summations of vectors (and thus
linear sketches) do not behave well in Hamming space.

Let $\mathcal{H}_n^k$ denote the set of $n$-dimensional binary vectors with $k$ ones. Let $\mathcal{H}_k(n)$ denote the set of $k$-dimensional vectors with
entries from $[n]$. We show that $\mathcal{H}_n^k$ can be embedded
into $\mathcal{H}_k(n)$ with distortion $O(\log n \log k)$.
\begin{restatable}{theorem}{thmembedding}
  There exists a map $\phi : \mathcal{H}_n^k \rightarrow \mathcal{H}_k(n)$ such that, for every $\vec{x}, \vec{y} \in \mathcal{H}_n^k$,
  \[\ham(\vec{x}, \vec{y}) / 2 \le \ham(\phi(\vec{x}), \phi(\vec{y})) \le O(\log n \log k) \ham(\vec{x}, \vec{y}).\]
  \label{thm:embedding}
\end{restatable}

The densification embedding is a simple application of the
worker-task assignment problem. In order to embed a vector
$\vec{x} \in \mathcal{H}_n^k$ into $\mathcal{H}_k(n)$, we simply
assign the workers $\{1, 2, \ldots, k\}$ to the task set
$T = \{i \mid \vec{x}_i = 1\}$, and then we construct the vector
$\vec{y}$ whose $i$-th coordinate denotes the task in $T$ to which
worker $i$ is assigned. This map transforms the switching cost in the
worker-task assignment problem into the distortion of the metric
embedding. 

The densification embedding is optimal in two senses. First, the
target space of the embedding must have $\Omega(k)$ coordinates simply
in order to allow for distances of $\Omega(k)$. Second, when
$k \ll n$, any embedding of $\mathcal{H}_n^k$ to $k$-dimensional
Hamming space must use $\Omega(\log n)$ bits per coordinate, simply in
order so that the embedding is an injection. It is not clear whether
the distortion achieved by our embedding is optimal, however, and it
remains open whether smaller distortion can be achieved by allowing
for a larger target-space dimension.

We remark that the basic relationship between worker-task assignment and densification embeddings problem has already implicitly been observed in previous work  on lower bounds, as a way to formalize what makes the worker-task assignment problem difficult \cite{su2020lower}. In contrast, here we are using the relationship as an avenue to obtain improved upper bounds for the densification problem. 

\section{Technical overview}

This section gives an overview of the main technical ideas in the
paper. For simplicity, the section will treat the task multiset
$T \subseteq [t]$ as always being a set (rather than a
multiset). As discussed in Section \ref{sec:ub}, one can formally reduce from the multiset case to the set case, at the cost of $t$ being replaced with $t' = wt$.

\subsection{A warmup: The random-permutation algorithm}
We begin by describing a simple assignment function that we call the
\defn{random-permutation algorithm}. The random-permutation algorithm
does not necessarily achieve small switching cost, but it does have the property
that for any two adjacent task sets $T_1, T_2 \subseteq [t]$, the
switching cost between $T_1$ and $T_2$ is $O(\log w)$ with high
probability in $w$. 

\paragraph{The algorithm} The random-permutation algorithm assigns to
each worker $i \in [w]$ a random permutation
\[\sigma_i = \langle \sigma_i(1), \sigma_i(2), \ldots, \sigma_i(t)
\rangle\] of the numbers $[t]$. We think of worker $i$ as
\defn{preferring} task $\sigma_i(j)$ over task $\sigma_i(j + 1)$ for
all $j \in [t - 1]$.

Suppose we wish to assign workers to tasks $T$.  The
random-permutation algorithm assigns the workers $1, 2, \ldots, w$ to
tasks $\tau_1, \tau_2, \ldots, \tau_w \in T$ one by one in order of worker ID, assigning worker $i$ to the task that it most prefers out of the tasks in $T$ that have not yet been assigned a worker.

For each $i \in [w]$, we define the \defn{$i$-remainder tasks} to be
the tasks $T \setminus \{\tau_1, \ldots, \tau_i\}$. That is, the
$i$-remainder tasks are the tasks that remain after the first $i$
workers are assigned. This means that worker $i + 1$ is assigned to
the $i$-remainder task that it most prefers.

\paragraph{Analyzing expected switching cost}
Let $T_1, T_2 \subseteq [t]$ be adjacent task sets of size $w$. We
begin by showing that the expected switching cost from $T_1$ to $T_2$
is $O(\log w)$.

Let $r$ and $s$ be such that $T_1 = (T_2 \cup \{r\}) \setminus \{s\}$. Let
$\psi_1$ and $\psi_2$ denote the assignments produced by the random
permutation algorithm for $T_1$ and $T_2$, respectively. Let $A_i$ and
$B_i$ denote the set of $i$-remainder tasks during the constructions
of $\psi_1$ and $\psi_2$, respectively.

The key to analyzing the random-permutation algorithm is to compare
the $i$-remainder sets $A_i$ and $B_i$ for each $i \in [w]$. We claim
that $|A_i \setminus B_i| \le 1$ for all $i \in [w]$. We prove this by
induction on $i$: suppose that
$A_{i - 1} = (B_{i - 1} \cup \{a\}) \setminus \{b\}$, and suppose for
contradiction that $|A_i \setminus B_i| \ge 2$. If either $\psi_1$
assigns worker $i$ to task $a$, or $\psi_2$ assigns worker $i$ to task
$b$, then we would be guaranteed that $|A_i \setminus B_i| \le 1$, a
contradiction. Thus $\psi_1$ and $\psi_2$ must each assign worker $i$
to a task in $A_{i - 1} \cap B_{i - 1}$. But this means that, in both
assignments, worker $i$ is assigned to the task in
$A_{i - 1} \cap B_{i - 1}$ that worker $i$ most prefers. Thus $\psi_1$
and $\psi_2$ assign worker $i$ to the same task, again contradicting
that $|A_{i} \setminus B_i| \ge 2$.

We now analyze the probability of $\psi_1$ and $\psi_2$ differing in
their assignment of worker $i$. Since $A_i$ contains at most one
element $a$ not in $B_i$, the probability that worker $i$ prefers $a$
over all elements in $B_i$ is at most $1 / |B_i| = 1 / (w - i +
1)$. Similarly, since $B_i$ contains at most one element $b$ not in
$A_i$, the probability that worker $i$ prefers $b$ over all elements
in $A_i$ is at most $1 / |A_i| = 1 / (w - i + 1)$. By the union bound,
it follows that the probability of $\psi_1$ and $\psi_2$ assigning
worker $i$ to different tasks is at most $2 / (w - i + 1)$.

By linearity of expectation, the expected switching cost between $T_1$
and $T_2$ is at most
\[\sum_{i = 1}^w \frac{2}{w - i + 1} = O(\log w).\]

\paragraph{Why a union bound fails for worst-case switching cost}
By using Chernoff-style bounds, one can modify the above analysis of
the random-permutation algorithm to show that, with high probability
in $w$ (i.e., probability $1 - 1/\poly w$), the switching cost between
$T_1$ and $T_2$ is $O(\log w)$. 

On the other hand, achieving a switching cost of $O(\log w)$ for \emph{all} pairs $(T_1, T_2)$ of adjacent task sets presents a challenge because
there are $\binom{w + 1}{2} \binom{t}{w + 1}$ such pairs that
must be considered. When $w = t / 2$, the number of distinct pairs
$(T_1, T_2)$ of adjacent task sets exceeds $2^t \ge 2^w$.

Thus, the probability bounds achieved by the random-permutation
algorithm are nowhere near high enough to enable a union bound over
all adjacent worker-set pairs. We call this the \defn{union-bound
  magnitude issue}.

\subsection{An algorithm with small switching cost}

We now describe a randomized assignment algorithm $\mathcal{A}$ that,
with high probability in $t$, achieves switching cost
$O(\log w \log t)$ on \emph{all} adjacent task-sets
$T_1, T_2 \subseteq [t]$ of size $w$. That is, with high probability, $\mathcal{A}$ produces an assignment function satisfying the requirements of Theorem \ref{thm:UB}.
 The algorithm $\mathcal{A}$ is called the \defn{multi-round balls-to-bins algorithm}. 

The multi-round balls-to-bins algorithm essentially flips the approach
taken by the random-permutation algorithm. One can think of the
random-permutation algorithm as consisting of $w$ phases in which each
phase \emph{deterministically} assigns exactly one worker to a task, and then
the phases \emph{probablistically} incur small switching cost. In
contrast, the multi-round balls-to-bins algorithm consists of
$O(\log w)$ phases, where each phase \emph{probabilistically} assigns
some number of workers to tasks, and each phase
\emph{deterministically} incurs small switching cost. Whereas the
failure mode of the random-permutation algorithm is that a
high-switching cost may occur, the failure mode of the multi-round
balls-to-bins algorithm is that some workers may be left
unassigned. As we shall see later, this distinction plays an important
role in solving the union-bound magnitude issue.

\paragraph{Structure of the multi-round balls-to-bins algorithm}
We begin with a succinct description of the algorithm
$\mathcal{A}$. For each $i$ from $1$ to $\log_{1.1} w$, repeat the
following hashing procedure $c\log t$ many times. Initialize a hash
table consisting of $w/(1.1)^i$ bins and randomly hash each unassigned
worker and each unassigned task into this table. For each bin that
contains at least one worker and one task, assign the minimum worker
in that bin to the minimum task in that bin.

In more detail, the algorithm $\mathcal{A}$ is the composition of $O(\log w)$
sub-algorithms $\mathcal{A}_1, \mathcal{A}_2, \ldots$.  Each of
$\mathcal{A}_1, \mathcal{A}_2, \ldots$ are \defn{partial-assignment}
algorithms, meaning that $\mathcal{A}_i$ assigns some subset of the
workers to some subset of the tasks in $T$, possibly leaving workers
and tasks unassigned. Note that the input to algorithm $\mathcal{A}_i$
is the set of workers/tasks that remain unassigned by
$\mathcal{A}_1, \ldots, \mathcal{A}_{i - 1}$.  Thus one can think of
the input to $\mathcal{A}_i$ as being a pair $(W, T)$ where
$W \subseteq [w]$ is a set of workers, $T \subseteq [t]$ is a set of
tasks, and $|W| = |T|$.

The algorithm $\mathcal{A}_1$'s responsibility is to assign enough
workers to tasks so that at most $w / 1.1$ workers remain
unassigned. Algorithm $\mathcal{A}_2$ is then executed on the
remaining (i.e., not-yet-assigned) workers and tasks, and is
responsible for assigning enough workers to tasks so that at most
$w / (1.1)^2$ workers remain unassigned. Continuing like this,
algorithm $\mathcal{A}_i$ is executed on the workers/tasks that remain
unassigned by all of $\mathcal{A}_1, \ldots, \mathcal{A}_{i - 1}$, and
is responsible for assigning enough workers to tasks that at most
$r_i = w / (1.1)^i$ workers in $W$ remain unassigned.

Each of the $\mathcal{A}_i$'s are randomized algorithms, meaning that
they have some probability of failure. The failure mode for
$\mathcal{A}_i$ is \emph{not} high-switching cost, however. In fact,
as we shall see later, each $\mathcal{A}_i$ deterministically
contributes at most $O(\log w)$ to the switching cost. Instead, the
way in which $\mathcal{A}_i$ can fail is that it may leave more than
$r_i$ workers unassigned. This means that the failure mode for the
full algorithm $\mathcal{A}$ is that it may fail to assign all of the
workers in $W$ to tasks in $T$. 

\paragraph{Applying the probabilistic method to
  $\mathcal{A}_1, \mathcal{A}_2, \ldots$}
Before describing the partial-assignment algorithms $\mathcal{A}_i$ in
detail, we first describe how our analysis of algorithm $\mathcal{A}$
overcomes the union-bound magnitude issue.

Recall that each algorithm $\mathcal{A}_i$ is responsible for reducing
the number of remaining workers to $r_i = w / (1.1)^i$. We will later
see that each $\mathcal{A}_i$ has a failure probability $p_i$ that is
a function of $r_i$ and $t$, namely,
\[p_i = \frac{1}{t^{\Omega(r_i)}}.\]
As $i$ grows, the failure probability $p_i$ of $\mathcal{A}_i$ becomes
larger, making it impossible to union-bound over exponentially many
pairs of task sets $T_1, T_2$.

An important insight is that, if all of
$\mathcal{A}_1, \ldots, \mathcal{A}_{i - 1}$ succeed (i.e., they each
assign the number of workers that they are responsible for assigning)
then the number of workers and tasks that $\mathcal{A}_{i - 1}$ is
executed on is only $O(r_i)$. That is, if we think of the inputs to
$\mathcal{A}_i$ as being pairs $(W, T)$ where $W \subseteq [w]$ is a
set of workers and $T \subseteq [t]$ is a set of tasks, the set of
inputs $(W, T)$ that algorithm $\mathcal{A}_{i - 1}$ must succeed on
is only the inputs for which $|W| = |T| \le O(r_i)$. The number of
such inputs is at most $t^{O(r_i)}$. In other words, even though the
failure probability $p_i$ of algorithm $\mathcal{A}_i$ increases with
$i$, the number of inputs over which we must apply a union bound
decreases. By a union bound, we can deduce that $\mathcal{A}_i$ has a
high probability in $t$ of succeeding on all relevant inputs $(W,
T)$. Combining this analysis over all of the partial-assigning
algorithms $\mathcal{A}_1, \mathcal{A}_2, \ldots$, we get that the
full assignment algorithm $\mathcal{A}$ also succeeds with high
probability in $t$. In particular, we have proven that there exists a deterministic assignment function with the desired switching cost, and that such a function can be obtained with high probability by the randomized algorithm $\mathcal{A}$.

\paragraph{Designing $\mathcal{A}_i$}
Each algorithm $\mathcal{A}_i$ is a composition of $\Theta(\log t)$
algorithms
$\mathcal{A}_{i, 1}, \mathcal{A}_{i, 2}, \mathcal{A}_{i, 3}, \ldots$,
each of which individually is a partial assignment algorithm.

Each algorithm $\mathcal{A}_{i, j}$ takes a simple balls-in-bins
approach to assigning some subset of the remaining workers to some
subset of the remaining tasks.

In particular, $\mathcal{A}_{i, j}$ places the workers into bins
$1, 2, \ldots, r_i$ by hashing each worker to a bin (using a random function from $[w]$ to $[r_i]$). Similarly, the
tasks are placed into bins $1, 2, \ldots, r_i$ by hashing each task to
a bin. If a bin $b$ contains both at least one worker and at least one
task, then the smallest-numbered worker in bin $b$ is assigned to the
smallest-number task in bin $b$.

Note that each of the algorithms
$\mathcal{A}_{i, 1}, \mathcal{A}_{i, 2}, \mathcal{A}_{i, 3}, \ldots$
are identical copies of one-another, except using different random
bits. Also note all of the $\mathcal{A}_i$'s are defined in the same way as each other, except the number of bins hashed to decreases as $i$ increases.  As we shall see shortly, the reason for having $\mathcal{A}_i$
consist of $\Theta(\log t)$ sub-algorithms is to enable probability
amplification later in the analysis.

\paragraph{Bounding the switching cost}
The partial assignment algorithms $\mathcal{A}_{i, j}$ are designed to
satisfy two essential properties, which we prove formally in the full proof. These two properties can then be combined to bound
the switching cost of the full algorithm $\mathcal{A}$. 

\begin{description}
\item[Compatibility:] Let $I_1 = (W_1, T_1)$ and
  $I_2 = (W_2, T_2)$ be inputs to $\mathcal{A}_{i, j}$. Suppose $I_1$
  and $I_2$ are \defn{unit distance}, meaning that
    \[|W_1\setminus W_2|+|W_2 \setminus W_1| + |T_1\setminus T_2|+|T_2 \setminus T_1| \le 2.\]
  Let $I_1' = (W_1', T_1')$ and $I_2' = (W_2', T_2')$ be the workers
  and tasks that remain unassigned when $\mathcal{A}_{i, j}$ is
  executed on each of $I_1$ and $I_2$, respectively. Then $I_1'$ and
  $I_2'$ are guaranteed to also be unit-distance.
\item[Low Switching Cost:] The switching cost of
  $\mathcal{A}_{i, j}$ is $O(1)$. That is, if $I_1 = (W_1, T_2)$ and
  $I_2 = (W_2, T_2)$ are inputs to $\mathcal{A}_{i, j}$, and $I_1$ and
  $I_2$ are unit-distance, then the worker-task assignments made by
  $\mathcal{A}_{i, j}$ on each of $I_1$ and $I_2$ differ by at most
  $O(1)$ assignments.
\end{description}

Consider two adjacent task sets $T_1$ and $T_2$. When we execute
$\mathcal{A}$ on $T_1$ and $T_2$, respectively, we use $I_1^{i, j}$
and $I_2^{i, j}$, respectively, to denote the worker/task input that
are given to partial-assignment algorithm $\mathcal{A}_{i,
  j}$.

The Compatibility property of the $\mathcal{A}_{i, j}$'s guarantees by
induction that, for each $\mathcal{A}_{i, j}$ the worker/task inputs
$I_1^{i, j}$ and $I_2^{i, j}$ are unit-distance (or
zero-distance). The Low-Switching-Cost property then guarantees that
each $\mathcal{A}_{i, j}$ contributes at most $O(1)$ to the switching
cost of $\mathcal{A}$. Since there are $O(\log t \log w)$
$\mathcal{A}_{i, j}$'s, this bounds the total switching cost of
$\mathcal{A}$ by $O(\log t \log w)$.

\paragraph{Deriving the success probabilities}
Next we analyze the probability of $\mathcal{A}_i$ failing on a given
worker/task input $(W, T)$. Recall that the only way in which
$\mathcal{A}_i$ might fail is if more than $r_i$ workers remain
unassigned after $\mathcal{A}_i$ finishes. Additionally, since we need
only consider cases where $\mathcal{A}_{i - 1}$ succeeds, we can
assume that $r_i \le |W|, |T| \le 1.1 r_i$.

Let $q$ denote the number of workers that $\mathcal{A}_{i, 1}$ assigns
to tasks. Given that $r_i \le |W|, |T| \le 1.1 r_i$, a simple analysis
of $\mathcal{A}_{i, 1}$ shows that $\E[q] \ge r_i / 5$. On the other
hand, using McDiarmid's inequality, one can perform a balls-in-bins
style analysis in order to show that
$\Pr[\E[q] - q > r_i / 10] \le 2^{-\Omega(r_i)}$. This means that
$\mathcal{A}_{i, 1}$ has probability at most $2^{-\Omega(r_i)}$ of
leaving more than $r_i$ workers unassigned.

In order for $\mathcal{A}_i$ to fail (i.e., $\mathcal{A}$ leaves more
than $r_i$ workers unassigned), \emph{all} of sub-algorithms
$\mathcal{A}_{i , 1}, \mathcal{A}_{i, 2}, \ldots$ would have to
individually fail. Since there are $\Theta(\log t)$ sub-algorithms,
the probability of them all failing is
\[p_i = 2^{-\Omega(r_i \log t)} = t^{-\Omega(r_i)}.\]
This allows us to apply the probabilistic method to the
$\mathcal{A}_i$'s in order to bound the probability of any
$\mathcal{A}_i$ failing on any input, as desired.

\paragraph{An explicit construction with polylogarithmic switching cost}
The multi-round balls-to-bins algorithm gives a non-explicit approach
to constructing an assignment function with low switching cost. The
approach is non-explicit because it relies on the probabilistic
method.

We also show how to obtain an explicit algorithm with switching cost
$\polylog wt$. The basic idea is to replace random hash functions,
used to place workers and tasks into bins,
with functions obtained from pseudorandom objects called
strong dispersers. Instead of trying a number of random hash functions
within the $\mathcal{A}_{i,j}$'s, we instead iterate over all of the hash functions from a small family given by a strong disperser~\cite{Meka14}.

\subsection{A lower bound on switching cost}

Define $s_{w,t}$ to be the optimal switching cost for assignment functions that assign workers $1, 2, \ldots, w$ to multisets of $w$ tasks from the universe $[t]$. The upper bounds in this paper establish that $s_{w, t} \le O(\log w \log (wt))$. It is natural to wonder whether smaller bounds can be achieved, and in particular, whether a small switching cost that depends only on $w$ can be achieved.

It trivially holds that $s_{w, t} \le w$. We show that when $t$ is sufficiently large relative to $w$, there is a matching lower bound of $s_{w, t} \ge w$. In fact, our lower bound only uses the evaluation of the assignment function on sets (as opposed to multisets).

Consider an assignment function $\phi$ that, given a multiset $T$ of tasks with elements from $[t]$ of $w$ tasks, produces an assignment of workers $[w]$ to tasks $T$. Our goal will be to find tasks $\tau_1<\tau_2<\cdots<\tau_{w+1}$ such that if $\phi(\{\tau_1,\ldots,\tau_{w}\})$ assigns worker $i$ to task $\tau_{\pi(i)}$ for some permutation $\pi$ of $[w]$, then $\phi(\{\tau_2,\ldots,\tau_{w+1}\})$ assigns worker $i$ to task $\tau_{\pi(i)+1}$. The existence of such a configuration immediately implies that $\phi$ has switching cost $w$.

We use an application of the hypergraph Ramsey theorem to show that, when $t$ is large enough, a configuration of the type described in the above paragraph must exist. Let $K_t^{(w)}$  denote the complete $w$-uniform hypergraph on $t$ vertices. This is just the set of $w$-element subsets of $[t]$, which correspond to sets of tasks. For each hyperedge $T=\{\tau_1,\ldots,\tau_w\}$, where $1\leq \tau_1<\cdots<\tau_w\leq t$, we color the hyperedge $T$ by a color $\pi$ where $\tau_{\pi(i)}$ is the task assigned to worker $i$.

This gives a coloring of the hyperedges of $K_t^{(w)}$ by $w!$ colors, each color being a permutation of $[w]$. By the hypergraph Ramsey theorem, if $t$ is large enough in terms of $w$, there must exist $w+1$ vertices $\tau_1,\ldots,\tau_{w+1}$ so all the hyperedges formed by the vertices have the same color $\pi$. By examining the hyperedges $\{\tau_1, \ldots, \tau_w\}$ and $\{\tau_2, \ldots, \tau_{w + 1}\}$, it follows that $\phi(\{\tau_1,\ldots,\tau_{w}\})$ assigns each worker $i$ to task $\tau_{\pi(i)}$ and that $\phi(\{\tau_2,\ldots,\tau_{w+1}\})$ assigns each worker $i$ to task $\tau_{\pi(i)+1}$, as desired.

\section{Achieving switching cost \texorpdfstring{$O(\log w \log (wt))$}{O(log w log (wt))}}\label{sec:ub}

In this section, we prove the following theorem.

\thmUB*

We demonstrate the existence of such a function via the probabilistic method, showing that there is a randomized construction that produces a low-switching cost worker-task assignment function with nonzero probability. In Section~\ref{sec:derandomizing} we also show how to derandomize the construction at the cost of a few extra log factors.

\paragraph{From multisets to sets}
We begin by showing that, without loss of generality, we can restrict our attention to task multisets $T$ that are sets (rather than multisets). We reduce from the multiset version of the problem with $w$ workers and $t$ tasks to the set version of the problem with $w$ workers and $wt$ tasks.

\begin{lemma}
  Define $n = wt$. Let $\phi$ be a worker-task assignment function that
  assigns workers $[w]$ to task sets $T \subseteq[n]$ (note that $\phi$ is
  defined only on task \emph{sets} $T$, and not on multisets).  Let $s$ be
  the switching cost of $\phi$ (considering only pairs of adjacent
  subsets of $[n]$, rather than adjacent sub-multisets). Then there
  exists a worker-task assignment function $\phi'$ assigning workers
  $[w]$ to task multisets $T \subseteq [t]$, such that $\phi'$ also has
  switching cost $s$.
  \label{lem:multiset_reduction}
\end{lemma}
\begin{proof}
  When discussing the assignment function $\phi$, we think of its input task-set $T$ as consisting of elements from $[t] \times [w]$ rather than elements of $[tw]$.
  
  With this in mind, we construct $\phi'$  as follows. Given a task multiset $T \subseteq [t]$, define the set $\si(T) \subseteq [t] \times [w]$ to be $\bigcup_{i = 1}^t \big\{(i,1), \ldots, (i, m_T(i))\big\}$, where $m_T(i)$ is the multiplicity of $i$ in $T$. The worker-task assignment $\phi$ produces some bijection $\psi_{\si(T)}: [w] \to \si(T)$. Similarly, $\phi'$ should produce some bijection $\psi'_T:[w] \to T$. This bijection is defined naturally by projection: if $\psi_{\si(T)}$ assigns worker $j$ to task $(i, x)$, let $\psi'_T$ assign worker $j$ to task $i$. 

  We now compute the switching cost of $\phi'$. Let $T$ and $T'$ be two adjacent task multisets, so $T' = T \cup \{a\} \setminus \{b\}$ for some $a,b \in [t]$. Then $\si(T') = \si(T) \cup \{(a, m_T(a) + 1)\} \setminus \{(b, m_T(b))\}$, and so $\si(T')$ is adjacent to $\si(T)$. Since $\phi$ has switching cost $s$, $\psi_{\si(T)}$ and $\psi_{\si(T')}$ agree on $w - s$ workers. By construction, $\psi'_T$ and $\psi'_{T'}$ must agree on these $w-s$ workers as well, and so it too has switching cost at most $s$.
\end{proof}

In the remainder of the section, we will make the assumption that $T$
is a subset of $[n]$, and we will show how to design an assignment
function with switching cost $O(\log w \log n)$ on all pairs of
adjacent subsets of $[n]$. By Lemma \ref{lem:multiset_reduction},
setting $n = wt$ then implies Theorem \ref{thm:UB}.

\paragraph{Designing an assignment function as an algorithm} It will be helpful to think of the function we construct for assigning workers to tasks as an algorithm $\mathcal{A}$, which we call the \defn{multi-round balls-to-bins algorithm}. The algorithm $\mathcal{A}$ takes as input a set $T \subseteq [n]$ of tasks with $|T| = w$ and must produce a bijection from the workers $[w]$ to $T$.

The algorithm constructs this bijection in stages. Each stage is what we call a \defn{partial assignment algorithm}, which takes as input the current sets of workers and tasks that have yet to be matched and assigns some subset of these workers to some subset of the tasks. Formally, we define a partial assignment algorithm to be any function $\psi$ which accepts as input any pair of sets $T \subseteq [n],W \subseteq [w]$ with $|T| = |W|$ and produces a matching between some subset of $T$ and some subset of $W$. After applying $\psi$ to $(T,W)$, there may remain some unmatched elements $T' \subseteq T$, $W' \subseteq W$. We call $(T,W)$ the \defn{worker-task input} to $\psi$ and $(T',W')$ the \defn{worker-task output}. Since a matching must remove exactly as many elements from $T$ as it does from $W$, we must also have $|W'| = |T'|$. Consequently, there is a natural notion of the \defn{composition} of two partial assignment algorithms: the composition $\psi' \circ \psi$ applies $\psi$ and then $\psi'$, letting the worker-task output of $\psi$ be the worker-task input to $\psi'$.

\paragraph{The algorithm} 
We recall the description of the algorithm $\mathcal{A}$. For each $i$ from $1$ to $c\log w$, repeat the following hashing procedure $c\log n$ many times. Initialize a hash table consisting of $w/(1.1)^i$ bins and randomly hash each unassigned worker and each unassigned task into this table. For each bin that contains at least one worker and one task, assign the minimum worker in that bin to the minimum task in that bin. 

In more detail, our algorithm $\mathcal{A}$ is the composition of $\log_{1.1} w$ partial-assignment algorithms,
\[\mathcal{A} = \mathcal{A}_1 \circ \mathcal{A}_2 \circ \cdots \circ \mathcal{A}_{\log_{1.1} w}.\]
Let $c$ be a large positive constant. Each $\mathcal{A}_i$ is itself the composition of $c \log n$ partial-assignment algorithms,
	\[\mathcal{A}_i = \mathcal{A}_{i, 1} \circ \mathcal{A}_{i, 2} \circ \cdots \circ \mathcal{A}_{i, c \log n}.\]

\paragraph{Designing the parts} Each $\mathcal{A}_{i, j}$ assigns workers to tasks using what we call a \defn{$w / (1.1)^{i}$-bin hash}, which we define as follows.

For a given parameter $k$, a \defn{$k$-bin hash} selects functions $h_1:[w] \rightarrow [k]$ and $h_2:[n] \rightarrow [k]$ independently and uniformly at random. For each worker $\omega \in [w]$, we say that $\omega$ is \defn{assigned} to bin $h_1(\omega)$. Similarly, for each $\tau \in [n]$ we say $\tau$ is assigned to $h_2(\tau).$ These functions are then used to construct a partial assignment. Given a worker-task input $(W,T)$, we restrict our attention to only the assignments of workers in $W$ and tasks in $T$. In each bin $\kappa \in [k]$ with at least one worker and one task assigned, match the smallest such worker to the smallest such task. Importantly, once $h_1$ and $h_2$ are fixed, the algorithm $\mathcal{A}_{i, j}$ uses this same pair of hash functions for every worker-task input, which (as we will see later) is what allows it to make very similar assignments for similar inputs and achieve low switching cost. 

We set each $\mathcal{A}_{i, j}$ to be an independent random instance of the $k$-bin hash, where $k = w / (1.1)^{i}$. Formally, this means that the algorithm $\mathcal{A} = \mathcal{A}_{1, 1} \circ \cdots \circ \mathcal{A}_{\log_{1.1} w, c \log n}$ is a random variable whose value is a partial-assignment function. Our task is thus to prove that, with non-zero probability, $\mathcal{A}$ fully assigns all workers to tasks and has small switching cost.

\paragraph{Analyzing the algorithm}
In Section \ref{sec:bounding_switching_cost},
we show that $\mathcal{A}$ \emph{deterministically} has switching cost $O(\log w \log n)$.
	
Although $\mathcal{A}$ always has small switching cost, the algorithm is \emph{not} always a legal worker-task assignment function. This is because the algorithm may sometimes act as a \emph{partial} worker-task assignment function, leaving some workers and tasks unassigned.
	
In Section \ref{sec:bound_failure_prob}, we show that with probability greater than $0$ (and, in fact, with probability $1 - 1 / \poly n$), the algorithm $\mathcal{A}$ succeeds at fully assigning workers to tasks for \emph{all} worker-task inputs $(W, T)$. Theorem \ref{thm:UB} follows by the probabilistic method.

\subsection{Bounding the probability of failure}\label{sec:bound_failure_prob}

Call a partial-assignment algorithm $\psi$ \defn{fully-assigning} if for every worker/task input $(W, T)$, $\psi$ assigns all of the workers in $W$ to tasks in $T$. That is, $\psi$ never leaves workers unassigned.

\begin{proposition} The multi-round balls-to-bins algorithm $\mathcal{A}$ is fully-assigning with high probability in $n$. That is, for any polynomial $p(n)$, if the constant $c$ used to define $\mathcal{A}$ is sufficiently large, then $\mathcal{A}$ is fully-assigning with probability at least $1 - O(1/p(n))$.
 \label{prop:bound_failure_prob}
\end{proposition}

Proposition \ref{prop:bound_failure_prob} tells us that with high probability in $n$, $\mathcal{A}$ succeeds at assigning all workers on \emph{all} inputs. We remark that this is a much stronger statement than saying that $\mathcal{A}$ succeeds with high probability in $n$ on a \emph{given} input $(W, T)$.

The key to proving Proposition \ref{prop:bound_failure_prob} is to show that each $\mathcal{A}_i$ performs what we call \defn{$\left(w / (1.1)^i\right)$-halving}.  A partial-assignment function $\psi$ is said to perform \defn{$k$-halving} if for every worker/task input $(W, T)$ of size at most $1.1 k$, the worker-task output $(W', T')$ for $\psi(W, T)$ has size at most $k$. 

If every $\mathcal{A}_i$ performs $w / (1.1)^i$-halving, then it follows that
\[\mathcal{A}_1 \circ \cdots \circ \mathcal{A}_{\log_{1.1} w}\]
is a fully-assigning algorithm. Thus our task is to show that each $\mathcal{A}_i$ performs $w / (1.1)^i$-halving with high probability in $n$.

We begin by analyzing the $k$-bin hash on a given worker/task input $(W, T)$.

\begin{lemma} Let $\psi$ a randomly selected $k$-bin hash. Let $(W, T)$ be a worker/task input satisfying $|W| = |T| \le 1.1 k$, and let $(W', T')$ be the worker/task output of $\psi(W, T)$. The probability that $(W', T')$ has size $k$ or larger is $2^{-\Omega(k)}$.
  \label{lem:balls-to-bins-failure}
\end{lemma}
\begin{proof} We may assume that $|W|=|T|\geq k$, else the conclusion is trivially true. Let $X$ be the random variable denoting the number of worker/task assignments made by $\psi(W, T)$. Equivalently, $X$ counts the number of bins to which at least one worker is assigned and at least one task is assigned---call these the \defn{active bins}. We will show that $\Pr[X < \frac{k}{8}] \le 2^{-\Omega(k)}$. Since $|W| = |T| \leq 1.1k$, this immediately implies that $|W'| = |T'| \leq 1.1k-0.125k \leq k$ with probability $1-2^{-\Omega(k)}$, as desired.

  We begin by computing $\E[X]$. For each bin $j \in [k]$, the probability no workers are assigned to bin $j$ is $(1 - 1/k)^{|W|} \le (1 - 1/k)^{k} \le 1/e$. Similarly, the probability that no tasks are assigned to bin $j$ is at most $(1 - 1/k)^{|T|} \le 1/e$. The probability of bin $j$ being active is therefore at least $1 - 2/e \ge 1/4$. By linearity of expectation, $\E[X] \ge k / 4$.
  
  Next we show that the random variable $X$ is tightly concentrated around its mean. Because the bins that are active are not independent of one-another, we cannot apply a Chernoff bound. Instead, we employ McDiarmid's inequality:

  \begin{theorem}[McDiarmid '89 \cite{McDiarmid89}] Let $A_1, \ldots, A_m$ be independent random variables over an arbitrary probability space. Let $F$ be a function mapping $(A_1, \ldots, A_m)$ to $\mathbb{R}$, and suppose $F$ satisfies,
    \[\sup_{a_1, a_2, \ldots, a_m, \overline{a_i}} |F(a_1, a_2, \ldots, a_{i - 1}, a_i, a_{i + 1}, \ldots , a_m) - F(a_1, a_2, \ldots, a_{i - 1}, \overline{a_i}, a_{i + 1}, \ldots , a_m)| \le R,\]
    for all $1 \le i \le m$. That is, if $A_1, A_2, \ldots, A_{i - 1}, A_{i + 1}, \ldots, A_m$ are fixed, then the value of $A_i$ can affect the value of $F(A_1, \ldots, A_m)$ by at most $R$. Then for all $S > 0$,
    \[\Pr[|F(A_1, \ldots, A_m) - \E[F(A_1, \ldots, A_m)]| \ge R \cdot S] \le 2e^{-2S^2 / m}.\]
  \end{theorem}

  The number of active bins $X$ is a function of at most $2.2 \cdot k$ independent random variables (namely, the hashes $h_1(\omega)$ for each $\omega \in W$ and the hashes $h_2(\tau)$ for each $\tau \in T$). Each of these random variables can individually change the number of active bins by at most one. It follows that we can apply McDiarmid's inequality with $R = 1$ and $m = 2.2k$. Taking $S = k/8$, we obtain
  \[\Pr[|X - \E[X]| \ge k / 8] \le  e^{-\Omega(k)}.\]
  Since $\E[X] \ge k / 4$, we have that $\Pr[X < k / 8] \le e^{-\Omega(k)}$, which completes the proof of the lemma.
\end{proof}

Our next lemma shows that each $\mathcal{A}_i$ is $k$-halving with high probability in $n$, where $k = w / (1.1)^i$.

\begin{lemma} Let $\psi_1, \ldots, \psi_{c \log n}$ be independent random $k$-bin hashes, and let $\psi = \psi_{1} \circ \cdots \circ \psi_{c \log n}$. With high probability in $n$, $\psi$ is $k$-halving. That is, every worker-task input $(W,T)$ with $|W| = |T| \le 1.1k$ has a worker task output $(W', T')$ with $|W'| = |T'| \leq k$.
  \label{lem:k-halving}
\end{lemma}
\begin{proof} Fix an arbitrary worker-task input $(W,T)$ with $|W| = |T| \le 1.1 k$. Let $(W_i, T_i)$ denote the worker-task output after applying the first $i$ rounds, $\psi_1 \circ \cdots \circ \psi_i$. Let $p_i$ denote the probability that $ |W_i| = |T_i| > k$.

  First, we observe that $p_i \le e^{-\Omega(k)} p_{i-1}$ for all $i > 1$. Indeed, for $|W_i| = |T_i| > k$, we must necessarily have $ |W_{i-1}| = |T_{i-1}| > k$, which occurs with probability $p_{i-1}$, but in this situation, the probability that $\psi_i$ produces a worker-task output of size greater than $k$ is a further $e^{-\Omega(k)}$ by Lemma \ref{lem:balls-to-bins-failure}.

  The probability that $\psi$ fails to reduce the size of $(W,T)$ to $k$ or smaller is thus at most \begin{equation}
    p_{c \log n} \le e^{-\Omega(ck \log n)} \le n^{-\Omega(ck)},
    \label{eq:prob_fail}
  \end{equation}
  where $c$ is treated as a parameter.

  On the other hand, the number of possibilities for input pairs $(W,T)$ satisfying $|W| = |T| \le 1.1 k$ is 
  \begin{equation}
    \sum_{j = 0}^{1.1 k} \binom{w}{j}\binom{n}{j} \leq 1.1k \cdot w^{1.1k}n^{1.1k} \leq n^{O(k)}.
    \label{eq:num_options}
  \end{equation}
  Combining \eqref{eq:prob_fail} and \eqref{eq:num_options}, the probability that there exists \emph{any} pair $(W, T)$ of size $1.1 k $ or smaller which fails to have its size reduced to $k$ or smaller is at most $n^{O(k) - c\Omega(k)}$. If $c$ is selected to be a sufficiently large constant, then it follows that $\psi$ performs $k$-halving with probability at least $1 - n^{-\Omega(k)}$. \end{proof}

We now prove Proposition \ref{prop:bound_failure_prob}.

\begin{proof}[Proof of Proposition \ref{prop:bound_failure_prob}] By Lemma \ref{lem:k-halving}, each algorithm $\mathcal{A}_i$ is $\left(w / (1.1)^i\right)$-halving with high probability in $n$. By a union bound, it follows that all of $\mathcal{A}_i \in \{\mathcal{A}_1, \ldots, \mathcal{A}_{\log_{1.1} w}\}$ are $\left(w / (1.1)^i\right)$-halving with high probability in $n$. If this occurs, then
 \[\mathcal{A} = \mathcal{A}_1 \circ \cdots \circ \mathcal{A}_{\log_{1.1} w}\]
 is fully-assigning, as desired.
\end{proof}

\subsection{Bounding the switching cost}\label{sec:bounding_switching_cost}

Recall that two worker/task inputs $(W_1, T_1)$ and $(W_2, T_2)$ are said to be \defn{unit distance} if
    \[W_1\setminus W_2|+|W_2 \setminus W_1| + |T_1\setminus T_2|+|T_2 \setminus T_1| \le 2.\]
A partial-assignment algorithm $\psi$ is \defn{$s$-switching-cost bounded} if for all unit-distance pairs of worker/task inputs $(W_1, T_1)$ and $(W_2, T_2)$, the set of assignments made by $\psi(W_1, T_1)$ deterministically differs from the set of assignments made by $\psi(W_2, T_2)$ by at most $s$.

In this section, we prove the following proposition.
\begin{proposition} The multi-round balls-to-bins algorithm is $O(\log w \log n)$-switching-cost bounded.
  \label{prop:bounding_switching_cost}
\end{proposition}

We begin by showing that each of the algorithms $\mathcal{A}_{i, j}$ are $O(1)$-switching-cost bounded.

\begin{lemma} For any $k$, the $k$-bin hash algorithm is $O(1)$-switching-cost bounded.
  \label{lem:Aij_switching_cost}
\end{lemma}
\begin{proof} Let $\psi$ denote the $k$-bin hash algorithm. Consider unit-distance pairs of worker/task inputs $(W_1, T_1)$ and $(W_2, T_2)$. Changing $W_1$ to $W_2$ can change the assignments made by $\psi$ for at most a constant number of bins. Similarly changing $T_1$ to $T_2$ can change the assignments made by $\psi$ for at most a constant number of bins. Thus $\psi(W_1, T_1)$ differs from $\psi(W_2, T_2)$ by at most $O(1)$ assignments.
\end{proof}

Recall that $\mathcal{A}$ is the composition of the $O(\log w \log n)$ partial-assignment algorithms $\mathcal{A}_{i, j}$'s. The fact that each $\mathcal{A}_{i, j}$ is $O(1)$-switching-cost bounded does not directly imply that $\mathcal{A}$ is $O(\log w \log n)$-switching-cost bounded, however, because switching cost does not necessarily interact well with composition. In order to analyze $\mathcal{A}$, we show that each $\mathcal{A}_{i, j}$ satisfies an additional property that we call being composition-friendly.

A partial-assignment algorithm $\psi$ is \defn{composition-friendly}, if for all unit-distance pairs of worker/task inputs $(W_1, T_1)$ and $(W_2, T_2)$, the corresponding worker/task outputs $(W_1', T_1')$ and $(W_2', T_2')$ are also unit-distance.

Lemma \ref{lem:composition-friendly} shows that each $\mathcal{A}_{i, j}$ is composition-friendly.
\begin{lemma} For any $k$, the $k$-bin hash is composition-friendly.
  \label{lem:composition-friendly}
\end{lemma}
\begin{proof} 
  Although the algorithm $\psi$ is formally only defined on input $(W, T)$ for which $|W| = |T|$, we will abuse notation here and consider $\psi$ even on worker/task input $(W, T)$ satisfying $|W| \neq |T|$.\footnote{Indeed, the definition of the $k$-bin hash does not require a worker-task input with $|W| = |T|$. The only reason we require this equality in general is to simplify calculations, as in practice the algorithm will only be run on worker-task inputs of equal size.} Define the \defn{difference-score} of a pair of worker/task inputs  $I_1 = (W_1, T_1), I_2 = (W_2, T_2)$ to be the quantity
  \[d(I_1, I_2) = |W_1\setminus W_2|+|W_2 \setminus W_1| + |T_1\setminus T_2|+|T_2 \setminus T_1| .\] 
 
  We will show the stronger statement that the difference-score $d(O_1, O_2)$ of the corresponding worker/task outputs  $O_1 = (W_1', T_1'), O_2 = (W_2', T_2')$ satisfies
  \begin{equation}
  d(O_1, O_2) \le d(I_1, I_2).
    \label{eq:difference_score}
 \end{equation} 
  It suffices to consider only two special cases: the case in which $W_2 = W_1 \cup \{\omega\}$ for some worker $\omega$ and $T_2 = T_1$; and the case in which $T_2 = T_1 \cup \{\tau\}$ for some task $\tau$ and $W_2 = W_1$.
  Iteratively applying these two cases to transform $I_1$ into $I_2$ implies inequality \ref{eq:difference_score}. 
 
  For this purpose, the roles of $W$ and $T$ are identical, so suppose without loss of generality that $W_2 = W_1 \cup \{\omega\}$ for some worker $\omega$ and $T_2 = T_1$. Recall that the assignment of workers and tasks to buckets is determined by some hash functions $h_1,h_2$ and in particular is the same whether we input $W_1$ or $W_2$. We first assign (only) the elements of $W_1$ and $T_1$ to their respective buckets, and then look at how including the assignment of $\omega$ changes the worker-task output. If $h_1$ assigns $\omega$ to either a bin with no tasks or a bin which already has some lexicographically smaller worker, then we will have $W_2' = W_1' \cup \{w\}$ and $T_2' = T_1'$. If $h_1$ assigns worker $\omega$ to a bin with no other workers and at least one task, we let the smallest such task be $\tau$ and see  $W_2' = W_1'$ and $T_2' = T_1' \setminus \{\tau\}$. Finally, if $h_1$ assigns $\omega$ to a bin with only larger workers and at least one task, we let the minimal such worker be $\gamma$, and we see $W_2' = W_1' \cup \{\gamma\}$ and $T_2' = T_1'$. In all three cases, $d(O_1, O_2) = 1$, as desired.
\end{proof}

Next, we will show that composing composition-friendly algorithms has the effect of summing switching costs. 
\begin{lemma} Suppose that partial-assignment algorithms $\psi_1, \psi_2, \ldots, \psi_k$ are all composition-friendly, and that each $\psi_i$ is $s_i$-switching-cost bounded. Then $\psi_1 \circ \psi_2 \circ \cdots \circ \psi_k$ is composition-friendly and is $\left(\sum_i s_i\right)$-switching-cost-bounded.
\label{lem:concatenation}
\end{lemma}
\begin{proof} By induction, it suffices to prove the lemma for $k = 2$. Let $I_1 = (W_1, T_1)$ and $I_2 = (W_2, T_2)$ be unit-distance worker/task inputs.

  For $i \in \{1, 2\}$, let $I_i'= (W_i', T_i')$ be the worker/task output for $\psi_1(W_i, T_i)$, and let $I_i'' = (W_i'', T_i'')$ be the worker/task output for $\psi_2(W_i', T_i')$.

  Since $\psi_1$ is composition friendly, its outputs $I_1'$ and $I_2'$ are unit distance. Since $I_1'$ and $I_2'$ are unit distance, and since $\psi_2$ is composition friendly, the outputs $I_1''$ and $I_2''$ of $\psi_2$ are also unit distance. Thus $\psi_1 \circ \psi_2$ is composition friendly.

  Since the inputs $I_1$ and $I_2$ to $\psi_1$ are unit-distance, $\psi_1(I_1)$ and $\psi_1(I_2)$ differ in at most $s_1$ worker-task assignments. Since the inputs $I_1'$ and $I_2'$ to $\psi_2$ are also unit distance, $\psi_2(I_1')$ and $\psi_2(I_2')$ differ in at most $s_2$ worker-task assignments. Thus the composition $\psi_1 \circ \psi_2$ is $(s_1 + s_2)$-switching-cost bounded, as desired.
\end{proof}

We can now prove Proposition \ref{prop:bounding_switching_cost}.

\begin{proof}[Proof of Proposition \ref{prop:bounding_switching_cost}] By Lemma \ref{lem:Aij_switching_cost}, each $\mathcal{A}_{i, j}$ is $O(1)$-switching-cost bounded. By Lemma \ref{lem:composition-friendly}, each $\mathcal{A}_{i, j}$ is composition friendly. Since $\mathcal{A}$ is the composition of the $O(\log w \log n)$ different $\mathcal{A}_{i, j}$'s, it follows by Lemma \ref{lem:concatenation} that $\mathcal{A}$ is $O(\log w \log n)$-switching-cost bounded.
\end{proof}

\section{Derandomizing the construction}\label{sec:derandomizing}

In this section, we derandomize the multi-round balls-to-bins algorithm to prove the following theorem.

\thmUBcon*
To this end we use pseudorandom objects called \emph{strong dispersers}. 
Intuitively, a disperser is a function such that the image of any not-too-small subset of its large domain (e.g.,~workers or tasks) is a dense subset of its small co-domain (e.g.,~bins). Since this requirement is hard to satisfy directly,  dispersers are defined with a second argument, called the seed. For a strong disperser, the density requirement is satisfied only in expectation over the seed. The standard way to define strong dispersers (Definition~\ref{defn:strong} below) is in the language of random variables. We follow with an equivalent alternative Definition~\ref{defn:alt}, more convenient for our purposes.

\begin{definition}[Strong dispersers]\label{defn:strong}
For $k\in\NN$, $\epsilon\in\mathbb{R}_+$, a $(k, \epsilon)$-strong disperser is a function $Disp : \{0, 1\}^n \times \{0,1\}^d \to \{0,1\}^m$ such that for any random variable $X$ over $\{0,1\}^n$ with min-entropy at least $k$ we have 
\[|\operatorname{Supp}((Disp(X, U_d), U_d))| \geq (1 - \epsilon) \cdot 2^{m+d}.\]
\end{definition} 

Here $\operatorname{Supp}$ denotes the support of a random variable, $U_d$ denotes the uniform distribution on $\{0,1\}^d$, and the min-entropy of a random variable $X$ is defined as $\min_{x}(-\log_2(\Pr[X=x]))$. We will use a simple fact that any distribution which is uniform on a $2^k$-element subset of the universe and assigns zero probability elsewhere (called \emph{flat $k$-source} in pseudorandomness literature) has min-entropy $k$. Interestingly, every distribution with min-entropy at least $k$ is a convex combination of such distributions (see, e.g.,~Lemma~6.10 in~\cite{vadhan2012}, first proved in~\cite{chor1988}), which makes the following definition equivalent.

\begin{definition}[Strong dispersers, alternative definition]\label{defn:alt}
For $k\in\NN$, $\epsilon\in\mathbb{R}_+$, a $(k, \epsilon)$-strong disperser is a function $Disp : [N] \times [D] \to [M]$ such that for any subset $S \subseteq [N]$ of size $|S| \ge 2^k$ we have 
\[|\{(Disp(s,d),d) : s \in S, d \in [D]  \}| \geq (1 - \epsilon) \cdot M \cdot D.\]
\end{definition}

We use efficient explicit strong dispersers constructed by Meka, Reingold and Zhou~\cite{Meka14}.

\begin{theorem}[Theorem 6 in \cite{Meka14}]
\label{thm:disp}
For all $N=2^n$, $k\in\NN$, and $\epsilon\in\mathbb{R}_+$, there exists an explicit $(k, \epsilon)$-strong disperser
$Disp : [N] \times [D] \to [M]$ with $D=2^{O(\log n)}=\polylog N$ and $M=2^{k - 3\log n - O(1)}=2^k \cdot \Omega(1/\log^3 N)$.
\end{theorem}

\paragraph{Designing the algorithm}
We begin with applying Lemma~\ref{lem:multiset_reduction} in order to be able to restrict our attention to task sets (rather than multisets), at the expense of increasing the number of tasks from $t$ to $wt$. For convenience, we round up the new number of tasks to the closest power of two $N=2^{\lceil \log wt \rceil}$.

Our explicit algorithm $\mathcal{E}$ has the same structure as the randomized algorithm $\mathcal{A}$, i.e.~it is the composition of $\log w$ partial assignment algorithms 
\[\mathcal{E} = \mathcal{E}_1 \circ \mathcal{E}_2 \circ \cdots \circ \mathcal{E}_{\log w}.\]
Each $\mathcal{E}_i$ is responsible for bringing down the number of unassigned workers to the next power of two, and is composed of a number of explicit sub-algorithms $\mathcal{E}_{i,j}$'s. Contrary to $\mathcal{A}_{i,j}$'s, sub-algorithms $\mathcal{E}_{i,j}$'s are not identical copies for a fixed $i$. However, the chain of distinct sub-algorithms has to be copied $O(\log^3 N)$ times. We reflect this introducing the $\widehat{\mathcal{E}_i}$ notation:
\[\mathcal{E}_i = \underbrace{\widehat{\mathcal{E}}_i \circ \widehat{\mathcal{E}_i} \circ \cdots \circ \widehat{\mathcal{E}_i}}_{O(\log^3 N)\text{ times}},
\quad\text{where}\quad
\widehat{\mathcal{E}_i} = \mathcal{E}_{i,1} \circ \mathcal{E}_{i,2} \circ \cdots \circ \mathcal{E}_{i,\polylog N}.\]

The key difference between the randomized and explicit algorithm is that $\mathcal{E}_{i,j}$'s, instead of using random hash functions $h_1, h_2$, use explicit functions obtained from strong dispersers. Another notable difference is that $\mathcal{A}_{i,j}$'s use $k$ bins to deal with input sets of size in $[k, 1.1k]$, while $\mathcal{E}_{i,j}$'s have to use polylogarithmically less bins, limiting the number of worker-task pairs that can be assigned by a single sub-algorithm and, as a consequence, forcing us to compose a larger number of sub-algorithms.

Let us fix $i \in [\log w]$, and denote $k_i = \lceil\log w\rceil - i$. Let $Disp_i : [N] \times [D_i] \to [M_i]$ be the $(k_i, 1/4)$-strong disperser given by Theorem~\ref{thm:disp}. Recall that $D_i = \polylog N$, $M_i = 2^{k_i} \cdot \Omega(1/\log^3 N)$, and $N$ is large enough so that all workers and all tasks are elements of $[N]$. We will have $\widehat{\mathcal{E}_i} = \mathcal{E}_{i,1} \circ \mathcal{E}_{i,2} \circ \cdots \circ \mathcal{E}_{i,D_i}$.
 For each $j \in [D_i]$, sub-algorithm $\mathcal{E}_{i,j}$ assigns workers and tasks to $M_i$ bins. Each worker $\omega \in W$ is assigned to bin $Disp_i(\omega, j)$, and each task $\tau \in T$ is assigned to bin $Disp_i(\omega, j)$.
 Then, like in the randomized strategy, for each active bin (i.e.~one which was assigned nonempty sets of workers and tasks) the smallest worker and the smallest task in that bin get assigned to each other.

\paragraph{Analyzing switching cost}
In Section~\ref{sec:bounding_switching_cost}, where we analyze the switching cost of randomized multi-round balls-to-bins algorithm, we do not exploit the fact that the hash functions $h_1$, $h_2$ are random. Actually, as we already remark, our switching cost bound is deterministic and thus works for any choice of functions $h_1$, $h_2$. Therefore the same analysis works for the explicit algorithm.  Namely, each sub-algorithm $\mathcal{E}_{i,j}$ is $O(1)$-switching cost bounded and composition-friendly (Lemmas~\ref{lem:Aij_switching_cost} and~\ref{lem:composition-friendly} generalize trivially), thus the switching cost of $\mathcal{E}$ depends only on the number of sub-algorithms, which is $\polylog N = \polylog wt$, as desired.

\paragraph{Proving the algorithm is fully-assigning} We begin by analyzing the number of worker/task assignments made by $\widehat{\mathcal{E}_i} = \mathcal{E}_{i,1} \circ \cdots \circ \mathcal{E}_{i,D_i}$.

\begin{lemma}
\label{lem:Ehat}
Let $(W,T)$ be a worker/task input satisfying $|W|=|T| \ge 2^{k_i}$.
Then $\widehat{\mathcal{E}_i}(W,T)$ makes at least $M_i/4$ worker/task assignments.
\end{lemma}

\begin{proof}
By the definition of dispersers, the two images
\[\{(Disp_i(\omega,j),j) : \omega \in W, j \in [D_i] \}, \quad\text{and}\quad
  \{(Disp_i(\tau,  j),j) : \tau   \in T, j \in [D_i] \}\]
have size at least $(3/4) \cdot M_i \cdot D_i$. Since they are both subsets of $[M_i]\times[D_i]$, their intersection has size at least $(1/2) \cdot M_i \cdot D_i$. By the pigeonhole principle, there must exist $j \in D_i$ such that
\begin{equation}\label{eqn:disp}|Disp_i(W,j) \cap Disp_i(T,j)| \ge M_i / 2.\end{equation}
Let us fix such $j$, and look at the execution of $\mathcal{E}_{i,j}$. For each bin $b \in Disp_i(W,j) \cap Disp_i(T,j)$, if $b$ is not active, then all workers $\{\omega \in W \mid Disp_i(\omega,j)=b\}$ or all tasks \{$\tau \in T \mid Disp_i(\tau,j)=b\}$ must have been already assigned by $(\mathcal{E}_{i,1}\circ\cdots\circ\mathcal{E}_{i,j-1})(W, T)$. Thus, each bin in $Disp_i(W,j) \cap Disp_i(T,j)$ either is active -- and contributes one worker and one task to the assignment -- or is inactive and testifies that at least one worker or at least one task is assigned by earlier sub-algorithms.
Let $c_a$ denote the number of active bins, $c_w$~denote the number of inactive bins testifying for a worker assigned by earlier sub-algorithms, and $c_t$ denote the number of inactive bins testifying for a task. We have $c_a+c_w+c_t \ge M_i / 2$, by Inequality~(\ref{eqn:disp}). It follows that the number of worker/task assignments made by $(\mathcal{E}_{i,1}\circ\cdots\circ\mathcal{E}_{i,j})(W, T)$ is at least $c_a + \max(c_w, c_t) \ge c_a + \frac{1}{2}(c_w+c_t) \ge M_i / 4$, as desired.
\end{proof}

Recall that $M_i = 2^{k_i} \cdot \Omega(1/\log^3 N)$. Thus, Lemma~\ref{lem:Ehat} implies that each $\mathcal{E}_i$ -- which is a composition of $O(\log^3 N)$ copies of $\widehat{\mathcal{E}_i}$ -- when given a worker/task input of size at most $2 \cdot 2^{k_i}$ returns a worker/task output of size at most $2^{k_i}$. It follows that $\mathcal{E} = \mathcal{E}_1 \circ \mathcal{E}_2 \circ \cdots \circ \mathcal{E}_{\log w}$ is fully-assigning, which concludes the proof of Theorem~\ref{thm:UBcon}.

\section{Lower bounds on switching cost}

Define $s_{w,t}$ to be the optimal switching cost for assignment functions that assign workers $1, 2, \ldots, w$ to multisets of $w$ tasks from the universe $[t]$. The upper bounds in this paper establish that $s_{w, t} \le O(\log w \log (wt))$. It is natural to wonder whether smaller bounds can be achieved, and in particular, whether a small switching cost that depends only on $w$ can be achieved.

It trivially holds that $s_{w, t} \le w$. We show that when $t$ is sufficiently large relative to $w$, there is a matching lower bound of $s_{w, t} \ge w$.

\thmLB*

\begin{proof} Given any worker-task assignment function $\phi$, we can actually find high switching cost between a pair of task subsets, in which all demands are 0 or 1. For each $T\subseteq [t]$ of $w$ tasks, $\phi$ produces a bijection of workers $[w]$ to tasks $T$. In order to lower-bound the switching cost, we produce a coloring of the complete $w$-uniform hypergraph with $t$ vertices. The coloring will be designed so that, if it contains a monochromatic clique on $w + 1$ vertices, then the assignment function $\phi$ must have worst-possible switching cost $w$. By applying the hypergraph Ramsey theorem, we deduce that, if $t$ is large enough, then the coloring must contain a monochromatic $(w + 1)$-clique, completing the lower bound.

\paragraph{Coloring the complete $w$-uniform hypergraph on $t$ vertices}
Let $K_t^{(w)}$ denote the complete $w$-uniform hypergraph on $t$ vertices. Note that the hyperedges of $K_t^{(w)}$ are just the $w$-element subsets of $[t]$, which correspond to sets of tasks.

For a task set $T=\{\tau_1,\ldots,\tau_w\}$, where $1\leq \tau_1<\cdots<\tau_w\leq t$, we color the hyperedge $T$ with the tuple $\pi = \langle \pi(1), \pi(2), \ldots, \pi(w) \rangle$, where $\tau_{\pi(i)}$ is the task assigned to worker $i$. One can think of $\pi$ as a permutation of numbers $\{1,2,\ldots,w\}$, and thus the coloring consists of at most $w!$ colors.

\paragraph{Monochromatic $(w + 1)$-cliques imply high switching cost} The key property of the coloring $C$ is that, if $K_t^{(w)}$ contains a monochromatic $(w + 1)$-vertex clique (i.e., $K_{w + 1}^{(w)}$), then $\phi$ must have switching cost $w$.

Namely, if $K_t^{(w)}$ contains a monochromatic $(w + 1)$-clique, then we can find $w+1$ vertices, $\tau_1<\tau_2<\cdots<\tau_{w+1}$, such that every $w$-element subset $T$ of these tasks is assigned the same permutation $\pi$ as its color. In particular, this means that for the task-set $T_1 = \{\tau_1,\ldots,\tau_w\}$ each worker $i$ is assigned to task $\tau_{\pi(i)}$, but for the task-set $T_2 = \{\tau_2, \ldots, \tau_{w + 1}\}$ that same worker $i$ is assigned to a different task $\tau_{\pi(i) + 1}$. Thus there is a pair of adjacent task sets $T_1, T_2$ that exhibit switching cost $w$.

\subparagraph{Finding a monochromatic clique}
In order to complete the lower bound, we wish to show that, if $t$ is sufficiently large, then the coloring contains a monochromatic $K_{w + 1}^{(w)}$. To do this, we employ the hypergraph Ramsey theorem.

\begin{theorem}[{Theorem 1 in \cite{erdosrado1952combinatorial}}]
\label{thm:ramsey}
Let $k\geq 2$ and $N\geq n\geq 2$ be positive integers. The hypergraph Ramsey number $R(k,n,N)$ is defined to be the least positive integer $M$ such that for every $k$-coloring of the hyperedges of $K^{(n)}_M$, the complete $n$-uniform hypergraph on $M$ vertices, contains a monochromatic copy of $K^{(n)}_M$. This quantity satisfies  \[R(k,n,N)\leq k^{(k^{n-1})^{(k^{n-2})^{\cdots^{(k^2)^{k(N-n)+1}}}}}.\]
\end{theorem}

Applying Theorem \ref{thm:ramsey}, we see that if $t\geq R(w!,w,w+1)$, then the $(w!)$-coloring of $K_t^{(w)}$ contains a monochromatic $(w + 1)$-clique, and the assignment function $\phi$ must have switching cost $w$, as desired. By Theorem \ref{thm:ramsey}, $R(w!,w,w+1)\leq\tow(O(w))$.
which implies that that every worker-task assignment function has switching cost $w$ when $t\geq\tow(\Omega(w))$.
This completes the proof of Theorem \ref{thm:LB}.
\end{proof}

Another way of viewing this argument is that a worker-task assignment function with switching cost less than $w$ gives rise to a proper $(w!)$-coloring of a certain graph, with vertex set $\binom{[t]}{w}$ and edges of the form $(\{\tau_1,\ldots,\tau_w\}, \{\tau_2, \ldots, \tau_{w + 1}\})$ for $\tau_1<\tau_2<\cdots<\tau_{w+1}$. Such graphs are studied under the name of \defn{shift-graphs}, see, e.g., \cite[Section~3.4]{felsner1992interval}, where the definition and proofs of basic properties are attributed to~\cite{erdos1968chromatic}. In particular, the chromatic number of shift-graphs is known to be $(1+o(1))\cdot \log^{(w-1)} t$ (with the superscript denoting iteration).
This gives an alternative way to complete the proof of Theorem~\ref{thm:LB} and it gives the same asymptotic bound on $t$ in terms of $w$. While the chromatic number lower bound suffices to prove the switching cost bound, the nearly matching upper bound (on chromatic number) suggests that an entirely different technique would be needed in order to asymptotically improve the switching cost bound.

\section{Densification into Hamming space}\label{sec:metric}

In this section, we apply our results on worker-task assignment to the
problem of densification. In particular, we show how to embed sparse
high-dimensional binary vectors into dense low-dimensional Hamming
space.

Let $\mathcal{H}_n^k$ denote the set of $n$-dimensional binary vectors
with $k$ ones.  Let $\mathcal{H}_k(n)$ denote the set of
$k$-dimensional vectors with entries from $[n]$. We show that
$\mathcal{H}_n^k$ can be embedded into $\mathcal{H}_k(n)$ with
distortion $O(\log n \log k)$.

\thmembedding*
\begin{proof}
  Using Theorem \ref{thm:UB}, let $\psi$ be a worker-task assignment
  function mapping workers $1, 2, \ldots, k$ to a task set
  $T \subseteq [n]$ with switching-cost $O(\log n \log k)$.

  For $\vec{x} \in \mathcal{H}_n^k$, define
  $T(\vec{x}) = \{i \mid \vec{x}_i = 1\}$ to be the task set
  consisting of the positions in $\vec{x}$ that are $1$.  Define
  $\phi(\vec{x})$ to be the $k$-dimensional vector whose $i$-th
  coordinate denotes the task $t \in T(\vec{x})$ to which
  $\psi(T(\vec{x}))$ assigns worker $i$. For example, if $k = 3$,
  $\vec{x} = \langle 0, 1, 0, 1, 1, 0 \rangle$, and $\psi(T(\vec{x}))$
  assigns workers $1, 2, 3$ to tasks $4, 2, 5$, respectively, then
  $\phi(\vec{x}) = \langle 4, 2, 5 \rangle$.

  Since the coordinates of $\phi(\vec{x})$ are a permutation of the
  positions $T(\vec{x})$ in which $\vec{x}$ is non-zero, it is necessarily the
  case that
  \[\ham(\phi(\vec{x}), \phi(\vec{y})) \ge |T(\vec{x}) \setminus T(\vec{y})| \ge \ham(\vec{x}, \vec{y}) /
  2.\] On the other hand, since $\psi$ has switching cost
  $O(\log n \log k)$, it is also the case that $\psi(\vec{x})$ and $\psi(\vec{y})$
  differ by at most $O(\log n \log k) \ham(\vec{x}, \vec{y})$ assignments, meaning
  that,
  \[\ham(\phi(\vec{x}), \phi(\vec{y})) \le O(\log n \log k) \ham(\vec{x}, \vec{y}).\]
  This completes the proof of the theorem.
\end{proof}

  We remark that Theorem \ref{thm:embedding} can be generalized to
  allow for the the domain space $\mathcal{H}_n^k$ to have non-binary
  entries. In particular, if $\mathcal{L}_n^k$ is the set of vectors
  with non-negative integer entries that sum to $k$, then there is an
  embedding $\phi: \mathcal{L}_n^k \rightarrow \mathcal{H}_k(n)$ such
  that, for $x, y \in \mathcal{L}_n^k$,
   \[\ell_1(\vec{x}, \vec{y}) / 2 \le \ham(\phi(\vec{x}), \phi(\vec{y})) \le O(\log n \log k) \ell_1(\vec{x}, \vec{y}).\]
  This follows from the same argument as Theorem \ref{thm:embedding},
  except that now $T(x)$ is the \emph{multiset} for which each element
  $i \in [n]$ has multiplicity $\vec{x}_i$, and now $\psi$ is the
  worker-task assignment mapping workers $1, 2, \dots, k$ to a task
  \emph{multiset} $T \subseteq [n]$.

\section{Open problems}

We leave open the question of closing the gap between upper and lower bounds for the worker-task assignment problem: the upper bound is $\polylog(wt)$ and the lower bound is $\log^*(t)$.
 
One interesting parameter regime is when $w$ and $t$ are comparable in size (say within a polynomial factor of each other). In this regime, no super-constant lower bound is known. 
 
Another interesting direction is the problem of densification into Hamming space. Our upper bound for the worker-task assignment problem implies an upper bound for this problem, but our lower bound does not carry over. We leave open the problem of whether there is a better upper bound or a super-constant lower bound for this problem.
 
\bibliographystyle{plainurl}
\bibliography{writeup}

\begin{thebibliography}{10}

\bibitem{andoni2011near}
Alexandr Andoni, Moses~S Charikar, Ofer Neiman, and Huy~L Nguyen.
\newblock Near linear lower bound for dimension reduction in $\ell_1$.
\newblock In {\em 2011 IEEE 52nd Annual Symposium on Foundations of Computer
  Science}, pages 315--323. IEEE, 2011.

\bibitem{bansal2011polylogarithmic}
Nikhil Bansal, Niv Buchbinder, Aleksander Madry, and Joseph Naor.
\newblock A polylogarithmic-competitive algorithm for the $k$-server problem.
\newblock In {\em 2011 IEEE 52nd Annual Symposium on Foundations of Computer
  Science}, pages 267--276. IEEE, 2011.

\bibitem{berinde2008combining}
Radu Berinde, Anna~C Gilbert, Piotr Indyk, Howard Karloff, and Martin~J
  Strauss.
\newblock Combining geometry and combinatorics: A unified approach to sparse
  signal recovery.
\newblock In {\em 2008 46th Annual Allerton Conference on Communication,
  Control, and Computing}, pages 798--805. IEEE, 2008.

\bibitem{beshers2001models}
Samuel~N Beshers and Jennifer~H Fewell.
\newblock Models of division of labor in social insects.
\newblock {\em Annual review of entomology}, 46(1):413--440, 2001.

\bibitem{bourgain1985lipschitz}
Jean Bourgain.
\newblock On {L}ipschitz embedding of finite metric spaces in {H}ilbert space.
\newblock {\em Israel Journal of Mathematics}, 52(1-2):46--52, 1985.

\bibitem{brinkman2005impossibility}
Bo~Brinkman and Moses Charikar.
\newblock On the impossibility of dimension reduction in $\backslash$ell \_1.
\newblock In {\em Proceedings of the 44th Annual IEEE Symposium on Foundations
  of Computer Science}, page 514, 2003.

\bibitem{chakraborty2016streaming}
Diptarka Chakraborty, Elazar Goldenberg, and Michal Kouck{\`y}.
\newblock Streaming algorithms for embedding and computing edit distance in the
  low distance regime.
\newblock In {\em Proceedings of the 48th annual ACM Symposium on Theory of
  Computing}, pages 712--725, 2016.

\bibitem{charikar2018estimating}
Moses Charikar, Ofir Geri, Michael~P Kim, and William Kuszmaul.
\newblock On estimating edit distance: Alignment, dimension reduction, and
  embeddings.
\newblock In {\em 45th International Colloquium on Automata, Languages, and
  Programming (ICALP)}, volume 107, page~34, 2018.

\bibitem{charikar2006embedding}
Moses Charikar and Robert Krauthgamer.
\newblock Embedding the {U}lam metric into $\ell_1$.
\newblock {\em Theory of Computing}, 2(1):207--224, 2006.

\bibitem{charikar2002dimension}
Moses Charikar and Amit Sahai.
\newblock Dimension reduction in the $\ell_1$ norm.
\newblock In {\em The 43rd Annual IEEE Symposium on Foundations of Computer
  Science, 2002. Proceedings.}, pages 551--560. IEEE, 2002.

\bibitem{chor1988}
Benny Chor and Oded Goldreich.
\newblock Unbiased bits from sources of weak randomness and probabilistic
  communication complexity.
\newblock In {\em Proceedings of the 26th Annual Symposium on Foundations of
  Computer Science}, pages 429--442, 1985.

\bibitem{erdos1968chromatic}
P~Erd\H{o}s and A~Hajnal.
\newblock On chromatic number of infinite graphs.
\newblock In {\em Theory of Graphs (Proc. Colloq., Tihany, 1966)}, pages
  83--98. Academic Press, 1968.

\bibitem{erdosrado1952combinatorial}
Paul Erd\H{o}s and Richard Rado.
\newblock Combinatorial theorems on classifications of subsets of a given set.
\newblock {\em Proceedings of the London mathematical Society}, 3(1):417--439,
  1952.

\bibitem{fakcharoenphol2004tight}
Jittat Fakcharoenphol, Satish Rao, and Kunal Talwar.
\newblock A tight bound on approximating arbitrary metrics by tree metrics.
\newblock In {\em Proceedings of the thirty-fifth annual ACM Symposium on
  Theory of Computing}, pages 448--455, 2003.

\bibitem{felsner1992interval}
Stefan Felsner.
\newblock {\em Interval orders: combinatorial structure and algorithms}.
\newblock PhD thesis, Technische Universit{\"a}t Berlin, 1992.
\newblock URL: \url{http://page.math.tu-berlin.de/~felsner/Paper/diss.pdf}.

\bibitem{georgiou2011cooperative}
Chryssis Georgiou and Alexander~A Shvartsman.
\newblock Cooperative task-oriented computing: Algorithms and complexity.
\newblock {\em Synthesis Lectures on Distributed Computing Theory},
  2(2):1--167, 2011.

\bibitem{gilbert2006algorithmic}
A.~C. Gilbert, M.~J. Strauss, J.~A. Tropp, and R.~Vershynin.
\newblock Algorithmic linear dimension reduction in the $\ell_1$ norm for
  sparse vectors.
\newblock In {\em Allerton 2006 (44th Annual Allerton Conference on
  Communication, Control, and Computing}, 2006.

\bibitem{gilbert2010sparse}
Anna Gilbert and Piotr Indyk.
\newblock Sparse recovery using sparse matrices.
\newblock {\em Proceedings of the IEEE}, 98(6):937--947, 2010.

\bibitem{indyk2008explicit}
Piotr Indyk.
\newblock Explicit constructions for compressed sensing of sparse signals.
\newblock In {\em Proceedings of the nineteenth annual ACM-SIAM Symposium on
  Discrete Algorithms}, pages 30--33, 2008.

\bibitem{johnson1984extensions}
William~B Johnson and Joram Lindenstrauss.
\newblock Extensions of {L}ipschitz mappings into a {H}ilbert space.
\newblock {\em Contemporary mathematics}, 26(189-206):1, 1984.

\bibitem{krieger2000ant}
Michael~JB Krieger, Jean-Bernard Billeter, and Laurent Keller.
\newblock Ant-like task allocation and recruitment in cooperative robots.
\newblock {\em Nature}, 406(6799):992, 2000.

\bibitem{lerman2006analysis}
Kristina Lerman, Chris Jones, Aram Galstyan, and Maja~J Matari{\'c}.
\newblock Analysis of dynamic task allocation in multi-robot systems.
\newblock {\em The International Journal of Robotics Research}, 25(3):225--241,
  2006.

\bibitem{linial1995geometry}
Nathan Linial, Eran London, and Yuri Rabinovich.
\newblock The geometry of graphs and some of its algorithmic applications.
\newblock {\em Combinatorica}, 15(2):215--245, 1995.

\bibitem{macarthur2011distributed}
Kathryn~Sarah Macarthur, Ruben Stranders, Sarvapali~D Ramchurn, and Nicholas~R
  Jennings.
\newblock A distributed anytime algorithm for dynamic task allocation in
  multi-agent systems.
\newblock In {\em AAAI}, pages 701--706, 2011.

\bibitem{matouvsek2008variants}
Ji{\v{r}}{\'\i} Matou{\v{s}}ek.
\newblock On variants of the {J}ohnson--{L}indenstrauss lemma.
\newblock {\em Random Structures \& Algorithms}, 33(2):142--156, 2008.

\bibitem{McDiarmid89}
Colin McDiarmid.
\newblock On the method of bounded differences.
\newblock {\em Surveys in combinatorics}, 141(1):148--188, 1989.

\bibitem{mclurkin2005dynamic}
James McLurkin and Daniel Yamins.
\newblock Dynamic task assignment in robot swarms.
\newblock In {\em Robotics: Science and Systems}, volume~8. Citeseer, 2005.

\bibitem{mclurkin2004stupid}
James~Dwight McLurkin.
\newblock {\em Stupid robot tricks: A behavior-based distributed algorithm
  library for programming swarms of robots}.
\newblock PhD thesis, Massachusetts Institute of Technology, 2004.

\bibitem{Meka14}
Raghu Meka, Omer Reingold, and Yuan Zhou.
\newblock {Deterministic Coupon Collection and Better Strong Dispersers}.
\newblock In {\em Approximation, Randomization, and Combinatorial Optimization.
  Algorithms and Techniques (APPROX/RANDOM)}, volume~28 of {\em Leibniz
  International Proceedings in Informatics (LIPIcs)}, pages 872--884, 2014.
\newblock \href {https://doi.org/10.4230/LIPIcs.APPROX-RANDOM.2014.872}
  {\path{doi:10.4230/LIPIcs.APPROX-RANDOM.2014.872}}.

\bibitem{ostrovsky2007low}
Rafail Ostrovsky and Yuval Rabani.
\newblock Low distortion embeddings for edit distance.
\newblock {\em Journal of the ACM (JACM)}, 54(5):23--es, 2007.

\bibitem{radeva2017costs}
Tsvetomira Radeva, Anna Dornhaus, Nancy Lynch, Radhika Nagpal, and Hsin-Hao Su.
\newblock Costs of task allocation with local feedback: Effects of colony size
  and extra workers in social insects and other multi-agent systems.
\newblock {\em PLoS computational biology}, 13(12):e1005904, 2017.

\bibitem{robinson1992regulation}
Gene~E Robinson.
\newblock Regulation of division of labor in insect societies.
\newblock {\em Annual review of entomology}, 37(1):637--665, 1992.

\bibitem{csahin2004swarm}
Erol {\c{S}}ahin.
\newblock Swarm robotics: From sources of inspiration to domains of
  application.
\newblock In {\em International workshop on swarm robotics}, pages 10--20.
  Springer, 2004.

\bibitem{su2017ant}
Hsin-Hao Su, Lili Su, Anna Dornhaus, and Nancy Lynch.
\newblock Ant-inspired dynamic task allocation via gossiping.
\newblock In {\em International Symposium on Stabilization, Safety, and
  Security of Distributed Systems}, pages 157--171. Springer, 2017.

\bibitem{su2020lower}
Hsin-Hao Su and Nicole Wein.
\newblock {Lower Bounds for Dynamic Distributed Task Allocation}.
\newblock In {\em 47th International Colloquium on Automata, Languages, and
  Programming (ICALP)}, volume 168, pages 99:1--99:14, 2020.
\newblock \href {https://doi.org/10.4230/LIPIcs.ICALP.2020.99}
  {\path{doi:10.4230/LIPIcs.ICALP.2020.99}}.

\bibitem{vadhan2012}
Salil~P. Vadhan.
\newblock Pseudorandomness.
\newblock {\em Foundations and Trends in Theoretical Computer Science},
  7(1–3):1--336, 2012.
\newblock URL:
  \url{https://people.seas.harvard.edu/~salil/pseudorandomness/pseudorandomness-published-Dec12.pdf},
  \href {https://doi.org/10.1561/0400000010} {\path{doi:10.1561/0400000010}}.

\end{thebibliography}

\end{document}